\documentclass[11pt]{article}
\usepackage[T1]{fontenc}
\usepackage[utf8]{inputenc}
\usepackage{lmodern}
\usepackage{microtype}
\usepackage[a4paper,margin=1.1in]{geometry}
\usepackage[UKenglish]{babel}
\usepackage{url}
\usepackage[numbers,sort&compress]{natbib}
\usepackage{color}
\usepackage{doi}
\usepackage{latexsym}
\usepackage{mathrsfs}
\usepackage{amsmath}
\usepackage{amssymb}
\usepackage{amsthm}
\usepackage[all,warning]{onlyamsmath}
\usepackage{enumitem}
\usepackage{xparse}
\usepackage{mathtools}
\usepackage{graphicx}
\usepackage{complexity}
\usepackage[small]{caption}
\usepackage{hyperref}
\usepackage{euscript}

\theoremstyle{plain}
\newtheorem{theorem}{Theorem}
\newtheorem{lemma}[theorem]{Lemma}

\newtheorem{corollary}[theorem]{Corollary}

\theoremstyle{definition}

\theoremstyle{remark}

\setlist[itemize]{label=--}
\setlist[enumerate]{label=(\arabic*),labelindent=\parindent,leftmargin=*}

\DeclarePairedDelimiter\braces{\{}{\}}
\NewDocumentCommand\set{O{}mg}{\ensuremath{\braces[#1]{#2\IfNoValueTF{#3}{}{\,:\,#3}}}}

\newcommand{\N}{\mathbb{N}}

\newcommand{\lcl}{\textsf{LCL}}

\newclass{\local}{LOCAL}
\newclass{\congest}{CONGEST}

\newcommand{\namedref}[2]{\hyperref[#2]{#1~\ref*{#2}}}
\newcommand{\sectionref}[1]{\namedref{Section}{#1}}

\newcommand{\theoremref}[1]{\namedref{Theorem}{#1}}

\newcommand{\figureref}[1]{\namedref{Figure}{#1}}
\newcommand{\lemmaref}[1]{\namedref{Lemma}{#1}}

\newenvironment{myabstract}
               {\list{}{\listparindent 1.5em%
                        \itemindent    \listparindent
                        \leftmargin    1cm
                        \rightmargin   1cm
                        \parsep        0pt}%
                \item\relax}
               {\endlist}

\newenvironment{mycover}
               {\list{}{\listparindent 0pt
                        \itemindent    \listparindent
                        \leftmargin    1cm
                        \rightmargin   1cm
                        \parsep        0pt}%
                \raggedright
                \item\relax}
               {\endlist}

\newcommand{\myemail}[1]{\,$\cdot$\, {\small #1}\par\vspace{2pt}}
\newcommand{\myaff}[1]{{\small #1\par}\bigskip}

\newclass{\clique}{CLIQUE}
\newclass{\nclique}{NCLIQUE}
\newclass{\all}{ALL}
\newclass{\conclique}{coNCLIQUE}

\hypersetup{
  colorlinks=true,
  linkcolor=black,
  citecolor=black,
  filecolor=black,
  urlcolor=[rgb]{0,0.1,0.5},
  pdftitle={}, 
  pdfauthor={} 
}

\begin{document}

\mbox{}
\begin{mycover}
{\huge \bfseries Towards a complexity theory for the congested clique \par}
\bigskip
\bigskip

\textbf{Janne H.\ Korhonen}
\myemail{janne.h.korhonen@aalto.fi}
\myaff{Aalto University}

\textbf{Jukka Suomela}
\myemail{jukka.suomela@aalto.fi}
\myaff{Aalto University}

\end{mycover}

\bigskip
\begin{myabstract}
\noindent\textbf{Abstract.}
The \emph{congested clique} model of distributed computing has been receiving attention as a model for densely connected distributed systems. While there has been significant progress on the side of upper bounds, we have very little in terms of lower bounds for the congested clique; indeed, it is now known that proving explicit congested clique lower bounds is as difficult as proving circuit lower bounds.

In this work, we use various more traditional complexity theory tools to build a clearer picture of the complexity landscape of the congested clique:
\begin{itemize}
    \item \emph{Nondeterminism and beyond:} We introduce the \emph{nondeterministic congested clique model} (analogous to NP) and show that there is a natural canonical problem family that captures all problems solvable in constant time with nondeterministic algorithms. We further generalise these notions by introducing the \emph{constant-round decision hierarchy} (analogous to the polynomial hierarchy).
    \item \emph{Non-constructive lower bounds:} We lift the prior non-uniform counting arguments to a general technique for proving non-constructive uniform lower bounds for the congested clique. In particular, we prove a \emph{time hierarchy theorem} for the congested clique, showing that there are decision problems of essentially all complexities, both in the deterministic and nondeterministic settings.
    \item \emph{Fine-grained complexity:} We map out relationships between various natural problems in the congested clique model, arguing that a reduction-based complexity theory currently gives us a fairly good picture of the complexity landscape of the congested clique.
\end{itemize}
\end{myabstract}
\thispagestyle{empty}
\setcounter{page}{0}
\newpage

\section{Introduction}

\paragraph{\boldmath The congested clique.} In this work, we study computational complexity questions in the \emph{congested clique} model of distributed computing. The congested clique is essentially a fully-connected specialisation of the classic \congest{} model of distributed computing: There are $n$ nodes that communicate with each other in a fully-connected synchronous network by exchanging messages of size $O(\log n)$. Each node in the network corresponds to a node in an input graph $G$, each node starts with knowledge about their incident edges in $G$, and the task is to solve a graph problem related to~$G$.

The congested clique has recently been receiving increasing attention especially on the side of the upper bounds, and the fully-connected network topology enables significantly faster algorithms than what is possible in the \congest{} model. However, on the side of complexity theory, there has been significantly less development. Compared to the \local{} and \congest{} models, where complexity-theoretic results have generally taken the form of unconditional, explicit lower bounds for concrete problems, such developments have not been forthcoming in the congested clique. 
Indeed, it was show by Drucker et al.~\cite{drucker13} that congested clique lower bounds imply \emph{circuit} lower bounds, and the latter are notoriously difficult to prove -- overall, it seems that there are many parallels between computational complexity in the congested clique and \emph{centralised} computational complexity.

\paragraph{\boldmath Towards a complexity theory.} 
We use concepts and techniques from centralised complexity theory to map out the complexity landscape of the congested clique model. First, we focus on \emph{decision problems}:
\begin{itemize}
    \item We introduce the \emph{nondeterministic} version of the congested clique model. In particular, the class $\nclique(1)$ of problems solvable in constant time with nondeterministic algorithms is a natural analogue of the class NP. We show that there is a natural canonical problem family that captures all $\nclique(1)$ problems.
    \item We further generalise the notion of nondeterministic congested clique by introducing the \emph{constant-round decision hierarchy}, analogous to the polynomial hierarchy.
    \item We prove \emph{time hierarchy theorems} for the congested clique, showing that there are decision problems of essentially all complexities both in deterministic and nondeterministic settings.
\end{itemize}
Furthermore, we study the landscape of natural graph problems in the congested clique using a \emph{fine-grained complexity} approach:
\begin{itemize}
    \item While we cannot prove explicit lower bounds for the congested clique, we map out the \emph{relative complexity} of problems with polynomial complexity.
\end{itemize}

\subsection{Results: time hierarchy}

It is known that in the centralised setting, there are problems of almost any deterministic time complexity, due to the time hierarchy theorem~\cite{hartmanis1965computational,hennie1966two}. However, in distributed computing, we know that the picture can be quite different; for \lcl{} problems in the \local{} model, there are known \emph{complexity gaps}, implying that at least in some ranges \lcl{} problems can have only very specific complexities~\cite{naor95what,grid-lcl,local-hierarchy,lcl-complexity}. Thus, it makes sense to ask how the picture looks like in the congested clique: for example, it could be that -- similarly to \lcl{} problems in the \local{} model -- there are no problems with complexity $o(\log^* n)$ and $\omega(1)$.

We show that no such gaps occur in the congested clique model. Writing $\clique(T(n))$ for the set of decision problems that can be solved in $O(T(n))$ rounds, we prove a time hierarchy theorem for the congested clique: for any sensible complexity functions $S$ and $T$ with $S(n) = o(T(n))$, we have that 
\[ \clique(S(n)) \subsetneq \clique(T(n))\,.\]

The proof of the time hierarchy theorem is based on the earlier circuit counting arguments for a non-uniform version of the congested clique~\cite{doi:10.1142/S0129626416500043, drucker13}.
We show how to lift this result into the uniform setting, allowing us to show the existence of decision problems of essentially arbitrary complexity. Indeed, we use this same technique also for the other separation results in this paper.

\subsection{Results: nondeterminism and beyond}

\paragraph{\boldmath Nondeterministic congested clique.} The class NP and NP-complete problems are central in our understanding of centralised complexity theory. We build towards a similar theory for the congested clique by introducing a \emph{nondeterministic congested clique model}. We define the class $\nclique(T(n))$ as the class of decision problems that have nondeterministic algorithms with running time $O(T(n))$, or equivalently, as the set of decision problems $L$ for which there exists a deterministic algorithm $A$ that runs in $O(T(n))$ rounds and satisfies
\[ G \in L \quad \text{if and only if} \quad \exists z \colon A(G,z) = 1\,,\]
where $z$ is a \emph{labelling} assigning each node $v$ a nondeterministic guess $z_v$; for details, see \sectionref{sec:nondeterminism}.

We show that nondeterminism is only useful up to the number of bits communicated by the algorithm: any nondeterministic algorithm with running time $O(T(n))$ can be converted to a \emph{normal form} where each yes-instance has an accepting labelling with $|z_v| \le O(T(n) n \log n)$. As an application of this result, we show that $\nclique(S(n))$ does not contain $\clique(T(n))$ for any $S(n) = o(T(n))$.

\paragraph{\boldmath Constant-round nondeterministic decision.} We argue that the class $\nclique(1)$, consisting of problems solvable in constant time with nondeterministic algorithms, is a natural analogue of the class NP. The class $\nclique(1)$ contains most natural decision problems that have been studied in the congested clique, as well as many NP-complete problems such as $k$-colouring and Hamiltonian path. In particular, the question of proving that
\[ \clique(1) \ne \nclique(1) \]
can be seen as playing a role similar to the P vs.\ NP question in the centralised setting.
Alternatively, $\nclique(1)$ can be seen as an analogue of the class $\lcl$ of locally checkable labellings that has been studied extensively in the context of the \local{} model; see \sectionref{sec:conclusions}.

While we cannot prove a separation between deterministic and nondeterministic constant time, we identify a family of \emph{canonical problems} for $\nclique(1)$:
we show that any $\nclique(1)$ problem can be formulated as a specific type of \emph{edge labelling problem}. In particular, showing a non-constant lower bound for any edge labelling problem would be sufficient to separate $\clique(1)$ and $\nclique(1)$.

\paragraph{\boldmath Constant-round decision hierarchy.} We extend the notion of nondeterministic clique by studying a \emph{constant-round decision hierarchy}. This can be seen as analogous to the polynomial hierarchy in the centralised setting; each node can be seen as running an \emph{alternating} Turing machine.

Unlike for nondeterministic algorithms, it turns out that the label size for algorithms on the higher levels of this hierarchy is not bounded by the amount of communication. Thus, we get two very different versions of this hierarchy:
\begin{itemize}
    \item \emph{Unlimited hierarchy $(\Sigma_k, \Pi_k)_{k=1}^{\infty}$} with unlimited label size: we show that this version of the hierarchy collapses, as all decision problems are contained on the second level.
    \item \emph{Logarithmic hierarchy $(\Sigma^{\log}_k, \Pi^{\log}_k)_{k=1}^{\infty}$} with $O(n \log n)$-bit label per node: we show that there are problems that are not contained in this hierarchy.
\end{itemize}

\subsection{Results: fine-grained complexity}

By the time hierarchy theorem, we know that there are decision problems of all complexities, but it is beyond our current techniques to prove lower bounds for any specific problem, assuming we exclude lower bounds resulting from input or output sizes. However, what we can do is study the relative complexity of natural problems, much in the vein of centralised \emph{fine-grained} complexity. In \sectionref{sec:fine-grained}, we study the relative complexities of various concrete problems that are thought to have \emph{polynomial} complexity in the congested clique. Specifically, our framework is to compare \emph{problem exponents}, defined for problem $L$ as
\[ \delta(L) = \inf \{ \delta \in [0,1] \colon \text{$L$ can be solved in $O(n^\delta)$ rounds}\}\,. \]
By mapping out known relationships in this regime, we argue that despite the lack of explicit lower bounds, our understanding of the landscape of problems of polynomial complexity in the congested clique is in many senses better than e.g.\ in the \congest{} model.

\section{Related work}

\paragraph{\boldmath Upper bounds for the congested clique.} As noted in the introduction, upper bounds have been extensively studied in the congested clique model. Problems studied in prior work include routing and sorting~\cite{lenzen2013optimal}, minimum spanning trees~\cite{lotker05,Hegeman15_MST_logloglogn,logstarMST,korhonen-mst}, subgraph detection~\cite{tritri,censor2015algebraic}, shortest path problems~\cite{censor2015algebraic,becker2016near}, local problems~\cite{HegemanP14,hegeman14,censor2016derandomizing} and problems related to matrix multiplication~\cite{le2016further,censor2015algebraic}.

\paragraph{\boldmath Complexity theory for the congested clique.} Prior work on computational complexity in the congested clique is fairly limited; the notable exceptions are the connections to circuit complexity~\cite{drucker13} and counting arguments for the non-uniform version of the model \cite{doi:10.1142/S0129626416500043,drucker13}. However, lower bounds can be proven if we consider problems with large outputs; for example, lower bounds are known for triangle enumeration~\cite{pandurangan2016tight} or, trivially, a problem where all nodes are required to output the whole input graph.  Moreover, for the \emph{broadcast congested clique}, a version of the model where each node sends the same message to each other node every round, lower bounds have been proven using communication complexity arguments~\cite{drucker13}.

\paragraph{\boldmath Complexity theory for\,\,\congest{}.} For the \congest{} model, explicit lower bounds are known for many problems, even on graphs with very small diameter~\cite{PelegR-00,DHKNPPW-11,FHW-12,KP98,nanongkai14,LP13:podc}. These are generally based on reductions from known lower bounds in communication complexity; however, these reductions tend to boil down to constructing graphs with \emph{bottlenecks}, that is, graphs where large amounts of information have to be transmitted over a small cut. A key motivation for the study of the congested clique model is to understand computation in networks that do not have such bottlenecks.

\paragraph{\boldmath Complexity theory for\,\,\local{}.} Perhaps the most active development related to the computational complexity theory of distributed computing is currently taking place in the context of the \local{} model. There is a lot of very recent work that aims at developing a complete classification of the complexities of \lcl{} problems in the \local{} model \cite{chang16exponential,brandt16lll,ghaffari17distributed,grid-lcl,local-hierarchy,lcl-complexity}. In this line of research, the focus is on low-degree large-diameter graphs, while in the congested clique model we will study the opposite corner of the distributed computing landscape: high-degree low-diameter graphs.

\paragraph{\boldmath Nondeterminism and alternation.} Nondeterministic models of distributed computing have been studied under various names -- for example, \emph{proof labeling schemes} \cite{korman06distributed,korman07distributed,korman10proof,korman10constructing}, \emph{nondeterministic local decision} \cite{fraigniaud11ld}, and \emph{locally checkable proofs} \cite{goos16lcp} can be interpreted as nondeterministic versions of variants of the \local{} and \congest{} models; we refer to the survey by Feuilloley and Fraigniaud~\cite{decision-survey} for further discussion. However, there seem to be very few papers that take the next step from nondeterministic machines to alternating machines in the context of distributed computing -- we are only aware of Reiter~\cite{reiter2015distributed}, who studies alternating quantifiers in finite state machines, and Feuilloley et al.~\cite{feuilloley2016} and Balliu et al.~\cite{balliu_et_al:LIPIcs:2017:7025}, who study alternating quantifiers in the \local{} model.

\section{Preliminaries}

\paragraph{\boldmath The congested clique.} The congested clique is a specialisation of the standard \congest{} model of distributed computing to a fully connected network topology. The network consists of $n$ nodes (i.e.\ computers) that are connected to all other nodes by edges (i.e.\ communication links) -- that is, the communication graph is a clique.

As an input, we are given an undirected, unweighted graph $G = (V, E)$ with $V = \{ 1, 2, \dotsc, n \}$. Each node of the communication network has a unique identifier $v \in \{ 1,2, \dotsc, n \}$ and, in addition to it's own unique identifier, has initially knowledge about edges incident to node $v$ in $G$. The nodes collaborate to solve a problem related to the graph $G$.

The computation is done in synchronous rounds, and all nodes run the same deterministic algorithm. Each round, all nodes (1) perform an unlimited amount of local communication, (2) send a possibly different $O(\log n)$-bit message to each other node, and (3) receive the messages sent to them. The time complexity of an algorithm is measured in the number of rounds used. 

\paragraph{\boldmath Decision problems.} To avoid artefacts resulting from input and output sizes, we restrict our attention for the most part to decision problems on unweighted, undirected graphs. A \emph{decision problem} $L$ is a family of graphs; a graph $G$ is a yes-instance of $L$ if $G \in L$ and a no-instance otherwise. The \emph{complement} $\bar{L}$ of problem $L$ contains all graphs $G$ that are not in $L$.

Note that we do not require decision problems to be closed under isomorphisms, that is, problems can refer to the names of the nodes. However, we are only interested in decision problems that are computable in the centralised sense, and we will implicitly assume that this is the case for any problem considered.

An algorithm $A$ solves problem $L$ if for any graph $G$, each node $i$ produces output $A_i(G) = 1$ if $G \in L$ and $A_i(G) = 0$ if $G \notin L$; we write $A(G) = 1$ to indicate that all nodes produce output $1$ (the algorithm \emph{accepts}) and $A(G)= 0$ to indicate that all nodes output $0$ (the algorithm \emph{rejects}).
  
\paragraph{\boldmath Input encoding.} We will tacitly assume that the input is provided for node $v \in V$ in the form of a length-$(n-1)$ bit vector $x_v$ indexed by $V \setminus \{ v \}$ describing whether each of the potential incident edges is present in the input graph $G = (V,E)$. In particular, any two nodes $u$ and $v$ will share the bit $x_{u,v} = x_{v,u}$.

However, for technical reasons, it is convenient to consider a setting where each node has \emph{private} input bits, so we will implicitly assume that each of these bits is assigned to exactly one node, so that (1) for each possible edge in the graph, exactly one of its endpoints has the bit corresponding to that edge and (2) each node has at least $\lfloor (n-1)/2 \rfloor$ input bits. Note that it takes a single round to move from the latter setting to the former.

\paragraph{\boldmath Algorithms.} As noted above, we assume all nodes run the same deterministic algorithm. While the congested clique allows $O(\log n)$ bandwidth per round, where the constant hidden by O-notation can depend on the algorithm, we can always move the constant factors to the running time and assume that all algorithms use exactly $\lceil \log_2 n \rceil$ bits of communication per round.
 
\paragraph{\boldmath Deterministic complexity classes.} For a computable function \[T \colon \N \to \N\,,\] we define the complexity class $\clique(T(n))$ as the family of all graph problems that can be solved in $T(n)$ rounds.

\paragraph{\boldmath Counting arguments.} We will now review known results on the \emph{non-uniform} version of the congested clique~\cite{doi:10.1142/S0129626416500043,drucker13}. Specifically, we consider a setting where the number of nodes $n$ and the communication bandwidth $b$ is fixed beforehand, and we want to compute a function $f \colon \{ 0, 1 \}^{nL} \to \{ 0, 1 \}$, where $L$ is an integer; each node receives $L$ private input bits, and we want all nodes to output the same result $y \in \{ 0, 1 \}$. We define an \emph{$(n, b, L, t)$-protocol $P$} to be an algorithm that works in this setting and computes an output $P(x_1, x_2, \dotsc, x_n) \in \{ 0, 1 \}$ in $t$ rounds. Since the parameters are fixed, there are only a finite number of different $(n,b,L,T)$-protocols, and this number can be bounded by standard counting arguments:

\begin{lemma}[\cite{doi:10.1142/S0129626416500043}]\label{lemma:counting1}
The number of different $(n,b,L,t)$-protocols is at most 
\[ 2^{2bn^2 2^{L + bt(n-1)}}\,.\]
\end{lemma}

By contrast, the number of functions $f \colon \{ 0, 1 \}^{nL} \to \{ 0, 1 \}$ is $2^{2^{nL}}$, so this implies that for sufficiently large $n$, most such functions do not have a $(n,b,L,t)$-protocol when $t < L/b - 1$~\cite{doi:10.1142/S0129626416500043}.

\section{Time hierarchy}\label{sec:time-hierarchy}

We start by proving the deterministic time hierarchy theorem: there are problems of essentially all complexities in the congested clique model.

\begin{theorem}\label{thm:clique-hierarchy}
Let $S,T \colon \N \to \N$ be computable functions such that $S(n) = o(T(n))$ and $T(n) = O(n/\log n)$. Then 
\[ \clique(S(n)) \subsetneq \clique(T(n))\,.\]
\end{theorem}

\begin{proof}
We prove the theorem by constructing a language \[L \in \clique(T(n)) \setminus \clique(S(n)).\] For convenience, let us assume that $T(n) < n / (4\log n)$ for all sufficiently large $n$. Now, for all sufficiently large $n$, we define the set of $n$-node graphs that belong to $L$ as follows:
\begin{itemize}
    \item Fix $L = T(n) \log n \le \lfloor n/2 \rfloor$, and fix a function $f_n \colon \{ 0, 1 \}^{nL} \to \{ 0, 1 \}$ that does not have a $(n,\log n, L, T(n)/2)$-protocol; by \lemmaref{lemma:counting1}, such a function exists. Moreover, we can select $f_n$ to be a first function that satisfies this condition under the lexicographical ordering when interpreting functions $\{ 0, 1 \}^{nL} \to \{ 0, 1 \}$ as bit vectors of length $2^{nL}$.
    \item Let $G$ be a graph on $n$ nodes and let $x_v$ be the $L$-bit prefix of the input bit vector that node $v$ receives when $G$ is the input graph. We set $G \in L$ if $f_n(x_1, x_2, \dotsc, x_n) = 1$, and $G \notin L$ otherwise.
\end{itemize}

First, we observe that $L  \in \clique(T(n))$. That is, we can decide if the input graph $G$ belongs to $L$ in time $T(n)$ as follows:
\begin{enumerate}
    \item Each node $v$ broadcasts the first $L = T(n) \log n$ bits of its input -- that is, the vector $x_v$ -- to all other nodes. This takes $T(n)$ rounds.
    \item Each node $v$ uses local computation to find the function $f_n$ as specified above; this can be done by exhaustively enumerating all functions $f \colon \{ 0, 1 \}^{nL} \to \{ 0, 1 \}$ and all $(n,\log n,\allowbreak L, \allowbreak T(n))$-protocols. Each node then locally computes the value $f_n(x_1, x_2, \dotsc, x_n)$ and outputs it.
\end{enumerate}

It remains to show that $L \notin \clique(S(n))$. Assume for contradiction that there is an algorithm that solves $L$ in time $O(S(n))$. This implies that for any sufficiently large $n$, there is an $(n,\log n, L, \allowbreak O(S(n)))$-protocol $P_n$ for $f_n$. However, we have that $S(n) = o(T(n))$, so by the choice of $f_n$ the protocol $P_n$ cannot exist. 
\end{proof}

\section{Nondeterminism}\label{sec:nondeterminism}

\subsection{Nondeterministic complexity classes} A \emph{labelling $z$} of size $k$ is a mapping that assigns each node $v \in V$ a \emph{label $z_v \in \{ 0, 1 \}^*$} of length at most $k$.   A \emph{nondeterministic congested clique algorithm} $A$ is an algorithm that takes as an input, in addition to the input graph $G$, a labelling $z$ of size $S(n)$ for some computable function $S \colon \N \to \N$; we say that $S(n)$ is the labelling size of $A$. We can think of $z$ as the sequence of nondeterministic choices made by $A$, or alternatively as a certificate provided by an external prover. We say that $A$ decides the language $L$ if for all graphs $G$, 
\[ G \in L \quad \text{if and only if} \quad \exists z \colon A(G,z) = 1\,,\]
where $z$ is a labelling of size $S(n)$. 

For a computable function $T \colon \N \to \N$, we define the complexity class $\nclique(T(n))$ as the set of languages $L$ such that there exists a nondeterministic algorithm $A$ with running time of $T(n)$ rounds that decides $L$.

\subsection{\texorpdfstring{\boldmath $\nclique$}{NCLIQUE} normal form} While the definition of $\nclique(T(n))$ allows the algorithms to use an essentially arbitrary amount of nondeterministic bits, we show that nondeterministic bits are only useful, roughly speaking, as long as they can be communicated to other nodes by the algorithm. More precisely, we prove that any nondeterministic algorithm can be converted to a \emph{normal form}:

\begin{theorem}\label{thm:nclique-normal-form}
If $L \in \nclique(T(n))$, then there is a non\-de\-ter\-min\-is\-tic algorithm $B$ that decides $L$ with running time $T(n)$ and labelling size $O(T(n) n \log n )$.  
\end{theorem}

\begin{proof}
Let $A$ be the algorithm certifying $L \in \nclique(T(n))$. We say that a \emph{communication transcript} of an execution of $A$ of node $v$ is a bit vector consisting of all messages sent and received by $v$ during the execution of $A$. Clearly, a communication transcript of a node $v$ has length $O(T(n) n \log n)$.

We now define an algorithm $B$ that works as follows on input $(G,z)$:
\begin{enumerate}
    \item Each node $v \in V$ checks that their label $z_v$ is a valid communication transcript of length $O(T(n) n \log n )$ (if not, reject).
    \item Nodes verify that their labels are consistent with each other; this can be done in $T(n)$ rounds by simply replaying the transcripts and checking that all the received messages agree with the transcript (if not, reject).
    \item Each node $v \in V$ locally tries all possible local labels of size at most $S(n)$, where $S(n)$ is the labelling size of $A$, to see if there is a label $z'_v$ so that the execution of $A$ with local label $z'_v$ and the local input of node $v$ agrees with the transcript $z_v$ and accepts (if not, reject; otherwise accept).
\end{enumerate}
Clearly $B$ runs in $T(n)$ rounds. If there is a labelling $z'$ such that $A(G,z') = 1$, then using the transcripts from this execution of $A$ as the labelling $z$ clearly gives $B(G,z) = 1$. On the other hand, if there is a $z$ such that $B(G,z) = 1$, then there are local labels $z'_v$ for each $v \in V$ such that $A(G,z') = 1$.
\end{proof}

\subsection{Nondeterministic time hierarchy} As an application of the normal form theorem, we can extend the time hierarchy theorem to the nondeterministic congested clique. In fact, we prove a somewhat stronger statement:

\begin{theorem}\label{thm:nclique-hierarchy}
Let $S,T \colon \N \to \N$ be computable functions such that $S(n) = o(T(n))$ and $T(n) = O(n/\log n)$. Then there is a decision problem $L$ such that
\[ L \notin \nclique(S(n)) \quad \text{ and } \quad L \in \clique(T(n))\,.\]
\end{theorem}

\begin{proof}
We say that a $(n,b,M+L,t)$-protocol is a \emph{nondeterministic protocol} for function $f \colon \{ 0, 1 \}^{nL} \to \{ 0, 1 \}$ if for all $x \in \{ 0, 1 \}^{nL}$ it holds that $f(x_1, x_2, \dotsc, x_n) = 1$ if and only if there is $z \in \{ 0, 1 \}^{nM}$ such that $P(z_1x_1, z_2x_2, \dotsc, z_nx_n) = 1$.

Now we construct a decision problem $L \in \clique(T(n)) \setminus \nclique(S(n))$ using the same construction as in the proof of \theoremref{thm:clique-hierarchy}, with minor modifications as follows. Let $L = T(n) \log n$, and let $M = \frac{1}{4} T(n) n \log n$. We select the functions $f_n \colon \{ 0, 1 \}^{nL} \to \{ 0, 1 \}$ in the construction of $L$ with the extra constraint that $f_n$ does not have a nondeterministic $(n,\log n,M+L,T(n)/4)$-protocol; this is possible for sufficiently large $n$, since then
\[ M + L + T(n)(n-1) \log n \le \biggr(\frac{1}{2} + \frac{1}{n}\biggl)T(n) n \log n < \frac{3}{4} T(n) n \log n = \frac{3}{4} nL \,,\]


and thus by \lemmaref{lemma:counting1} the number of $(n,\log n,M+L,T(n)/4)$-protocols is $2^{o(2^{nL})}$.

The problem $L$ constructed using the functions $f_n$ is clearly in $\clique(T(n))$ using the same argument as in the proof of \theoremref{thm:clique-hierarchy}. On the other hand, if $L \in \nclique(S(n))$, then by \theoremref{thm:nclique-normal-form} there is a nondeterministic algorithm for $L$ with running time $O(S(n))$ and labelling size $O(S(n) n \log n)$. But this implies that the functions $f_n$ have nondeterministic $(n, \log n, L + O(S(n) n \log n),\allowbreak O(S(n)))$-protocols, which is not possible for large $n$ by the choice of $f_n$, since we have $S(n) = o(T(n))$.
\end{proof}

Since $\clique(T(n)) \subseteq \nclique(T(n))$, a time hierarchy theorem for the nondeterministic congested clique follows immediately from \theoremref{thm:nclique-hierarchy}.

\begin{corollary}\
Let $S,T \colon \N \to \N$ be computable functions such that $S(n) = o(T(n))$ and $T(n) = O(n/\log n)$. Then 
\[ \nclique(S(n)) \subsetneq \nclique(T(n))\,.\]
\end{corollary}

\section{Constant-round decision}\label{sec:constant-round}

\subsection{Constant-round nondeterministic clique} The class $\nclique(1)$ is a natural analogue of class NP in the congested clique; it contains decision versions of most natural problems considered in the congested clique setting. It is also easy to see that $\nclique(1)$ contains many NP-complete decision problems, such as $k$-colouring and Hamiltonian paths. By \theoremref{thm:nclique-hierarchy}, we also know that there are problems that can be solved in slightly super-constant time, but are not in $\nclique(1)$. However, our lower bound techniques are not sufficient to show that
\[ \clique(1) \ne \nclique(1)\,;\]
in a sense, this is the congested clique analogue of the P vs.\ NP question.

However, we can interpret \theoremref{thm:nclique-normal-form} to state that there is a \emph{canonical problem family} for $\nclique(1)$. Specifically, we say that a \emph{neighbourhood constraint $\mathcal{C}$} is a computable mapping that for any number of nodes $n$, any edge $\{u,v\}$, and any edge-neighbourhood $\partial(u)$ of $u$, gives a set $\mathcal{C}_{n,u,v,\delta(u)}$ of allowed $O(\log n)$-bit labels for the edge $\{u,v\}$. An \emph{edge labelling problem} is defined by an edge neighbourhood constraint $\mathcal{C}$: given a graph $G$, find a labelling $\ell$ of all edges of the clique with labels $\ell(u,v)$ of size $O(\log n)$ so that the labels satisfy the local constraints at all nodes, that is
\[ \ell(u,v) \in \mathcal{C}_{n,u,v,\partial(u)} \quad \text{ and } \quad \ell(u,v) \in \mathcal{C}_{n,v,u,\partial(v)} \]
for all $u$ and $v$.

By \theoremref{thm:nclique-normal-form}, any problem with $\nclique(1)$ algorithm $A$ can be interpreted as an edge labelling problem: the edge labels are defined as the set of valid communication transcripts of an accepting run of $A$. This gives us a limited notion of completeness for $\nclique(1)$:

\begin{theorem}
We have $\nclique(1) \subseteq \clique(T(n))$ if and only if all edge labelling problems can be solved deterministically in $O(T(n))$ rounds.
\end{theorem}
In particular, we have $\clique(1) = \nclique(1)$ if and only if all edge labelling problems can be solved deterministically in $O(1)$ rounds. However, it seems that identifying a single graph decision problem that is ``complete'' for $\nclique(1)$ is difficult; in essence, we would have to work reductions running in \emph{constant time}, which makes the use of any sort of gadget constructions extremely difficult.

\subsection{Constant-round decision hierarchy} Whereas $\nclique(1)$ is the congested clique analogue of NP, we can extend this analogue to the polynomial hierarchy by adding more quantifiers, that is, by allowing the nodes to alternate between nondeterministic and co-nondeterministic choices; similar ideas has been studied in the context of local verification~\cite{feuilloley2016,balliu_et_al:LIPIcs:2017:7025,reiter2015distributed}.

Formally, we say that a \emph{$k$-labelling algorithm $A$ of labelling size $S(n)$} is a constant-round congested clique algorithm that takes as an input $k$ labellings $z_1, z_2, \dotsc, z_k$ of size at most $S(n)$. We define the class $\Sigma_k$ as the set of languages $L$ for which there exists a $k$-labelling algorithm $A$ of labelling size $S(n)$ such that 
\[ G \in L \quad \text{ if and only if } \quad \exists z_1 \forall z_2 \dotso Q z_k \colon A(G,z_1, z_2, \dotsc, z_k) = 1\,,\]
where $z_1, z_2, \dotsc, z_k$ are labellings of size at most $S(n)$ and $Q$ is the universal quantifier if $k$ is even and the existential quantifier if $k$ is odd. Similarly, we define  $\Pi_k$ as the set of languages $L$ for which there exists a $k$-labelling algorithm $A$ of labelling size $S(n)$ such that 
\[ G \in L \quad \text{ if and only if } \quad \forall z_1 \exists z_2 \dotso Q z_k \colon A(G,z_1, z_2, \dotsc, z_k) = 1\,,\]
where $z_1, z_2, \dotsc, z_k$ are labellings of size at most $S(n)$ and $Q$ is the existential quantifier if $k$ is even and the universal quantifier if $k$ is odd. Finally, we define $\Delta_k = \Sigma_k \cup \Pi_k$.

At the first level of the hierarchy we have the class \[\Sigma_1 = \nclique(1)\,;\] by \theoremref{thm:nclique-normal-form} we know that limiting the labelling size to $O(n \log n)$ gives us the same class as unlimited labelling size. A natural question is then if the same phenomenon happens at the higher levels of the hierarchy?

Turns out this is not the case; we will consider two different versions of the constant-round decision hierarchy:
\begin{itemize}
    \item \emph{Unlimited hierarchy:} the hierarchy
    $(\Sigma_k, \Pi_k)_{k=1}^{\infty}$
     as defined above, with arbitrary labelling size allowed.
    \item \emph{Logarithmic hierarchy:} the hierarchy
    $(\Sigma^{\log}_k, \Pi^{\log}_k)_{k=1}^{\infty}$
    defined otherwise as above, but the algorithm $A$ is required to have labelling size of $O(n \log n)$ -- in other words, $O(\log n)$ bits per edge.
\end{itemize}
In a sense, these correspond to the different \local{} model hierarchies studied by Feuilloley et al.~\cite{feuilloley2016} (the logarithmic hierachy) and Balliu et al.~\cite{balliu_et_al:LIPIcs:2017:7025} (the unlimited hierarchy).

\paragraph{\boldmath Basic properties.} We first note basic properties of the constant-round hierarchies. Trivially, we have that
\[ \Sigma_k \subseteq \Delta_k \subseteq \Sigma_{k+1} \quad \text{ and } \quad \Pi_k \subseteq \Delta_k \subseteq \Pi_{k+1}\,,\]
and thus also
\[ \Pi_k \subseteq \Sigma_{k+1} \quad \text{ and } \quad \Sigma_k \subseteq \Pi_{k+1}\,.\]
Moreover, if a decision problem $L$ is in $\Sigma_k$, then the complement language $\bar{L}$ is in $\Pi_k$, and vice versa. These observations also hold for the logarithmic version of the constant-round hierarchy.

\paragraph{\boldmath  Unlimited hierarchy $(\Sigma_k, \Pi_k)_{k=1}^{\infty}$.} For the unlimited hierarchy, we obtain an essentially complete characterisation. By \theoremref{thm:nclique-hierarchy}, we know there are problems that are not on the first level. Moreover, in a similar manner as happens with the \local{} hierarchy of Balliu et al.~\cite{balliu_et_al:LIPIcs:2017:7025}, we show that the unlimited hierarchy collapses to the second level:

\begin{theorem}\label{thm:unlimited-hierarchy-collapse}
All decision problems $L$ are in $\Sigma_2 = \Pi_2$.
\end{theorem}

\begin{proof}
Let $L$ be a decision problem. To see that $L \in \Sigma_2$, consider the following algorithm $A$:
\begin{enumerate}
    \item The existential quantifier is used to guess a graph $G'_v$ in each node $v$, using $n^2$ bits per node.
    \item The universal quantifier is used to verify that all guesses $G'_v$ equal the input graph $G$: each node $v$ uses $O(\log n)$ universal bits to pick a single bit from the encoding of $G'_v$ and broadcasts this bit and its index to all other nodes. All nodes $u$ verify that the information broadcast by other nodes is consistent with their guess of $G'_u$ and their local view of $G$ (if not, reject).
    \item Each node $v$ locally checks if $G'_v \in L$ (if this holds for all nodes, accept; otherwise reject).
\end{enumerate}
Clearly, if the input graph $G$ is in $L$, then picking the real input graph $G$ as the existential guess results in $A$ accepting for all universal guesses. On the other hand, if $G \notin L$, then regardless of the existential guess, $A$ will not accept for all universal guesses: if $G'_v = G$ for all $v$, then step (3) will reject, and if $G'_v \ne G$ for some $v$, then there is a universal guess that will lead to step (2) rejecting.

Moreover, the above implies that for any decision problem $L$, we have that $\bar{L}\in \Sigma_2$, and thus $L \in \Pi_2$. It follows that all decision problems are also in $\Pi_2$.
\end{proof}

\paragraph{\boldmath  Logarithmic hierarchy $(\Sigma^{\log}_k, \Pi^{\log}_k)_{k=1}^{\infty}$.} For the logarithmic version of the constant-round hierarchy, the complexity landscape seems to be much richer that in the case of the unlimited hierarchy: any constant number of quantifiers is not enough to replicate the trick of \theoremref{thm:unlimited-hierarchy-collapse} of guessing the whole input at all nodes. Indeed, we show that there are problems that are not on any level of the logarithmic constant-round hierarchy:

\begin{theorem}\label{thm:log-hierarchy-lower-bound}
There is a decision problem $L$ such that
\[ L \notin \bigcup_{k = 0}^\infty \Sigma^{\log}_k\,.\]
\end{theorem}

\begin{proof}
Again, we use the same general proof technique as in the proofs of Theorems~\ref{thm:clique-hierarchy} and~\ref{thm:nclique-hierarchy}. Fix a computable function $T(n) = \omega(n)$; let $L = (T(n))^2 \log n$ and $M = \frac{1}{4} T(n) n \log n$. Otherwise using the same construction for the language $L$ as before, we select the functions $f_n \colon \{ 0, 1 \}^{nL} \to \{ 0, 1 \}$ so that for any $k \le T(n)$, there is no $(n,\log n, kM + L,(T(n))^2/4)$-protocol that $\Sigma^{\log}_{k}$-computes $f_n$; this is possible for sufficiently large $n$, since then
\[ kM + L + \frac{1}{4}(T(n))^2 (n-1) \log n < \frac{3}{4} (T(n))^2 n \log n = \frac{3}{4} nL \,,\]
and thus by \lemmaref{lemma:counting1} the number of $(n,\log n, kM + L,(T(n))^2/4)$-protocols is $2^{o(2^{nL})}$.

If $L \in \Sigma^{\log}_k$ for some $k$, then there is a $\Sigma_k$ algorithm for $L$ with running time $O(1)$ and labelling size $O(n \log n)$. But this implies that there is an $(n, \log n, O(kn \log n) + L, O(1))$-protocol that $\Sigma_k$-computes $L$; this is not possible for large $n$ by the choice of $f_n$, since we have $T(n) = \omega(1)$.
\end{proof}

The proof of \theoremref{thm:log-hierarchy-lower-bound} actually implies that there are problems with only slightly super-constant deterministic complexity that are not on any level of the hierarchy. However, our lower bound technique does not appear to be sufficiently refined to separate different levels of the logarithmic constant-round hierarchy, so it remains open whether the hierarchy has infinitely many levels.

\section{Fine-grained complexity}\label{sec:fine-grained}

\paragraph{\boldmath Problem exponent.} In the following, we will consider concrete problems that are not necessarily decision problems, and allow more generally problems defined in terms of weighted graphs and matrices. For a problem $L$, we define the \emph{exponent} of $L$ as
\[ \delta(L) = \inf \{ \delta \in [0,1] \colon \text{$L$ can be solved in $O(n^\delta)$ rounds}\}\,. \]
The basic idea is that the problem exponent captures the polynomial complexity of the problem, and we can compare the relative complexity of problems by comparing their exponents.

\begin{figure*}
\begin{center}
\includegraphics[width=\textwidth]{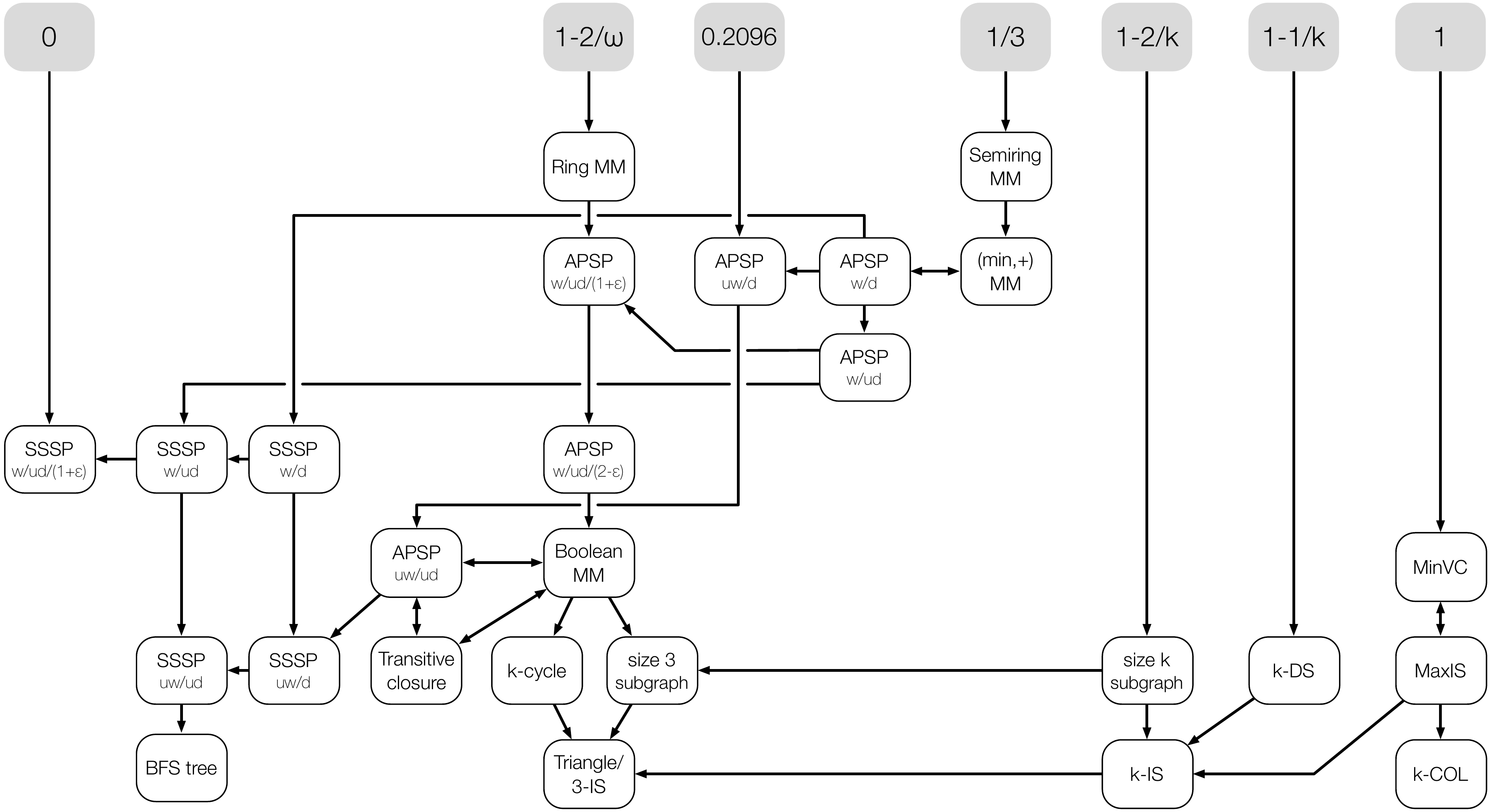}
\end{center}
\caption{Relationships between selected problems in the congested clique model. Arrow to $L_1$ from $L_2$ indicates $\delta(L_1) \le \delta(L_2)$; $k \ge 3$ and $\varepsilon > 0$ are arbitrary constants. For {\sf APSP} and {\sf SSSP}, {\sf w/uw} indicates weighted and unweighted variants, {\sf d/ud} indicates directed and undirected variants, and a number $\alpha$ indicates $\alpha$-approximate variant. Edge weights and matrix entries are assumed to be encodable in $O(\log n)$ bits in all problems.}\label{fig:complexity-ordering}
\end{figure*}

\paragraph{\boldmath Relationships between problems.}
In Figure 1, we summarise the relationships between prominent problems using this framework. These relations follow from prior work:
\begin{itemize}
    \item Relationships between $k$-independent set  ({\sf $k$-IS}), maximum independent set ({\sf MaxIS}) 
     and minimum vertex cover ({\sf MinVC}) are trivial. The upper bound for $k$-independent set is due to Dolev et al.~\cite{tritri}.
    \item Relationships between 3-independent set detection (i.e.\ triangle detection), $k$-cycle detection ({\sf $k$-CYCLE}), subgraph detection and Boolean matrix multiplication ({\sf Boolean-MM}) as well as between approximate weighted all-pairs shortest paths problem and ring matrix multiplication ({\sf Ring MM}) follow from the work of Censor-Hillel et al.~\cite{censor2015algebraic}; they also give the upper bound $\delta(\text{{\sf Ring MM}}) \le 1 - 2/\omega$, where $\omega < 2.3728639$ is the matrix multiplication exponent~\cite{legall2014powers}.
    \item The relationship between the $k$-colouring problem ({\sf $k$-COL}) and maximum independent set follows from an easy reduction~\cite{luby86}: replace each vertex $v$ with $k$ copies $v_1, \dotsc, v_k$ connected into a clique, and connect $v_i$ and $u_i$ if the edge $\{ v, u \}$ is present in the original graph. Clearly, the new graph has independent set of size $n$ if and only if the original graph is $k$-colourable. For constant $k$, the blowup from implementing this reduction is constant.
    \item The relationship between $(2-\varepsilon)$-approximate weighted undirected all-pairs shortest paths problem
    and Boolean matrix multiplication follows from a reduction by Dor et al.~\cite{doi:10.1137/S0097539797327908}.
    \item The relationships between other variants of the all-pairs shortest paths problem ({\sf APSP}) and single-source shortest paths problem ({\sf SSSP}) are trivial. The upper bound for unweighted directed APSP is due to Le Gall~\cite{le2016further}, and the upper bound for approximate SSSP is due to Becker et al.~\cite{transshipment}.
\end{itemize}
These are our main new contributions (see Sections \ref{sec:dom-ub}--\ref{sec:vc-ub} for proofs):
\begin{itemize}
    \item Dominating set of size $k$ ({\sf $k$-DS}) can be found in $O(n^{1-1/k})$ rounds, showing $\delta(\text{\sf $k$-DS}) \le 1-1/k$.
    \item For any fixed constant $k$, if a dominating set of size $k$ can be found in $O(n^{\delta})$ rounds, then an independent set of size $k$ can be found in $O(k^{2\delta+4}n^\delta)$ rounds, showing $\delta(\text{\sf $k$-IS}) \le \delta(\text{\sf $k$-DS})$.
    \item A vertex cover of size $k$ ({\sf $k$-VC}) can be found in $O(k)$ rounds, showing $\delta(\text{\sf $k$-VC}) = 0$ in our framework.
\end{itemize}
We note that one challenge involved in this approach is that the congested clique setting requires the use of extremely fine-grained reductions; we are essentially allowed only $n^{o(1)}$ factor blowups in the reductions. By contrast, most known reductions between NP-complete problems have polynomial blowup, making them useless in this setting.

One potentially fruitful perspective is to consider for which pairs of problems we \emph{cannot} prove reductions. For example, the reduction from Boolean matrix multiplication to $(2-\varepsilon)$-approximate APSP breaks down if we consider $2$-approximate APSP instead. Indeed, we know that constant-approximation APSP can be solved faster than the current matrix multiplication upper bound, using the spanner constructions of Censor-Hillel et al.~\cite{censor2016derandomizing}, so conjecturing the existence of a faster-than-matrix-multiplication algorithm for $2$-approximate APSP does not seem unreasonable. Another similar question is the existence of faster-than-matrix-multiplication algorithms for exact single-source shortest paths problems.

\subsection{Dominating set upper bound}\label{sec:dom-ub}

We now prove the upper bound for finding $k$-dominating sets:

\begin{theorem}
Dominating set of size $k$ can be found in $O(n^{1-1/k})$ rounds in the congested clique model.
\end{theorem}

We employ a slight modification of the Dolev et al.~\cite{tritri} algorithm for finding subgraphs of size~$k$. For convenience, let us assume that $n^{1/k}$ is an integer; if not, we use $\lfloor n^{1/k} \rfloor$ instead.
\begin{enumerate}
    \item Partition the node set $V$ arbitrarily into sets $S_1, S_2, \dotsc, S_{n^{1/k} }$ of size $O(n^{1 - 1/k})$.
    \item Assign each node $v \in V$ a label $\ell(v) \in [n^{1/k}]^k$, so that each possible label is assigned to some node. Moreover, we will assume that this is done in a globally consistent manner so all nodes will know labels of all nodes.
    \item For each node $v \in V$, let $S_v = S_{\ell(v)_1} \cup S_{\ell(v)_2} \cup \dotsb \cup S_{\ell(v)_k}$. Node $v$ learns all edges incident to nodes in $S_v$, and locally checks if there is a dominating set of size $k$ contained in $S_v$; clearly, knowing all edges incident to nodes in $S_v$ is sufficient for checking this.
\end{enumerate}
It remains to prove that the algorithm detects a dominating set of size $k$ if one exist, and that it runs in $O(kn^{1-1/k})$ rounds. The first part is simple: if there is a dominating set $D = \{ v_1, v_2, \dotsc, v_k \}$ of size $k$ such that $v_i \in S_{j_i}$, then some node $v \in V$ will receive the label $\ell(v) = (j_1, j_2, \dotsc, j_k)$ and detect $D$.

For the running time, we observe first that there are at most $k n^{2 - 1/k} = O(n^{2 - 1/k})$ edges incident to nodes in $S_v$ for all $v \in V$, so each node has to receive $O(n^{2 - 1/k})$ messages in step~3. On the other hand, we note that for each edge $e$, there are $2k n^{1-1/k}$ nodes $v \in V$ that have to learn about the existence of $e$, so each node has to send at most $2k n^{2-1/k}$ messages in step~3. Delivering the messages can thus be done in $O(n^{1-1/k})$ rounds using the routing protocol of Lenzen~\cite{lenzen2013optimal}.

\subsection{From independent set to dominating set}\label{sec:dom-lb}

We next prove that finding a $k$-dominating set is at least as hard as finding a $k$-independent set in the congested clique, for any fixed~$k$:

\begin{theorem}
If a dominating set of size $k$ can be found in $O(n^{\delta})$ rounds in the congested clique model, then an independent set of size $k$ can be found in $O(k^{2\delta+4}n^\delta)$ rounds.
\end{theorem}
 
The proof is straightforward; we construct a reduction from $k$\-independent set to $k$-dominating set that increases the number of vertices by a linear factor.

\paragraph{\boldmath Construction.}
Let $G = (V,E)$ be the input graph. We construct a new graph $G'=(V',E')$ such that $G$ has an independent set of size $k$ if and only if $G'$ has a dominating set of size $k$:
\begin{itemize}
    \item As a base for the construction of $G'$, we start with $k$ copies $K^1, K^2, \dotsc, K^k$ of clique on $n$ nodes. We identify the nodes in each clique with the original node set $V$ so that for each node $v \in V$, there is a corresponding node $v^i$ in the clique $K^i$.
    \item For each pair $(i,j) \in [k]^2$ with $i < j$, we construct a \emph{compatibility gadget} as follows. We add an independent set $I^{i,j}$ on $n$ nodes to $G'$, again identified with the original node set $n$ so that for each node $v \in V$, there is a corresponding node $v^{i,j}$ in the independent set $I^{i,j}$. We connect $I^{i,j}$ to $K^i$ and $K^j$:
    \begin{itemize}
        \item for each $v \in V$, we add an edge between the node $v^i$ in $K^i$ and $u^{i,j}$ in $I^{i,j}$ for all $u \in V \setminus \{ v \}$, and
        \item for each $v \in V$, we add an edge between the node $v^j$ in $K^j$ and $u^{i,j}$ in $I^{i,j}$ for all $u \in V \setminus \{ v \}$ that are not neighbours of $v$ in the original graph $G$.
    \end{itemize}
    \item For each clique $K^i$, we add two new \emph{special} nodes $x^i$ and $y^i$ that are connected to all nodes in $K^i$, but not to other nodes.
\end{itemize}
See \figureref{fig:reduction} for an illustration.

\begin{figure*}
\begin{center}
 \includegraphics[width=0.9\textwidth]{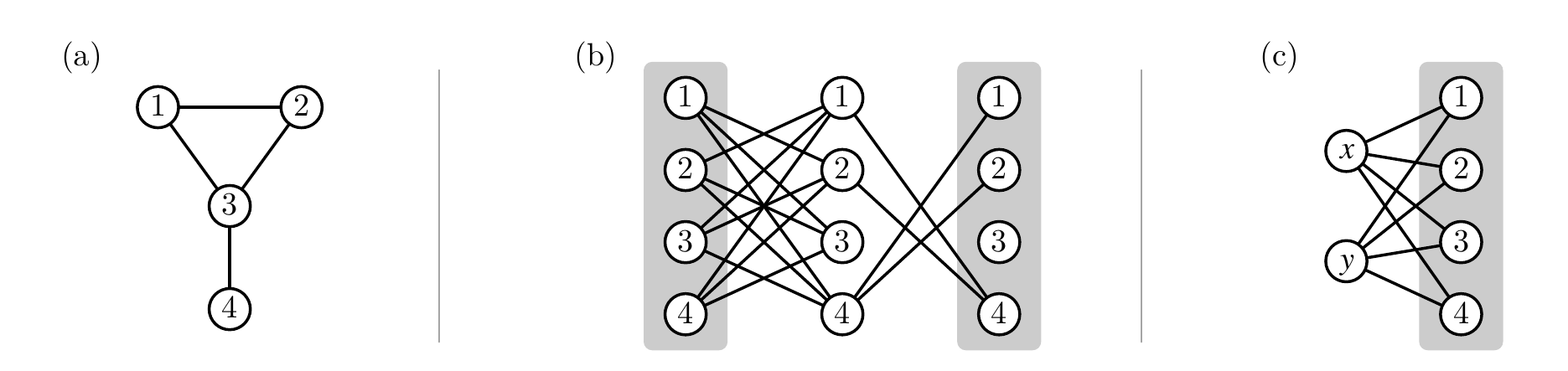}
\end{center}
    \caption{\label{fig:reduction}Gadgets in the reduction from independent set to dominating set. Corresponding nodes in the original graph and the gadgets are labelled with same numbers. (a)~The original graph~$G$. (b)~The compatibility gadget between two cliques, shown in gray. (c)~Special nodes attached to a clique.}
\end{figure*}

\paragraph{\boldmath Properties of the construction.} Clearly, the graph $G'$ has at most $(k^2 + k + 2)n$ nodes. To see the correspondence between independent sets in $G$ and dominating sets in $G'$, we make the following observations; note that we will tacitly abuse the identification between various nodes in the original graph $G$ and the new graph $G'$.
\begin{itemize}
    \item Consider an independent set $I = \{ v_1, v_2, \dotsc, v_k \}$ of size $k$ in $G$, and let $D$ be a node set obtained by picking $v_i^i$ from $K^i$ in $G'$ for all $i = 1,2,\dotsc, k$. Clearly $D$ dominates all nodes in set $K^i$ as well as the two special vertices attached to $K^i$, for all $i$. Furthermore, for each compatibility gadget corresponding to $(i,j)$, the node $v_i^i$ in $K^i$ dominates all nodes in $I^{i,j}$ except $v_i^{i,j}$; since $v_j$ is not a neighbour of $v_i$ in $G$, the node $v_j^j$ in $K^j$ dominates $v_i^{i,j}$ in $I^{i,j}$. Thus $D$ is a dominating set of size $k$ in $G'$.
    \item Consider a dominating set $D = \{ w_1, w_2, \dotsc, w_k \}$ of size $k$ in $G'$. First, we observe that $D$ must contain exactly one node from each of the cliques $K^1, K^2, \dotsc, K^k$, since otherwise all the special nodes are not dominated; for convenience, let us assume that $w_i$ is in $K^i$ for all $i$. Since $D$ is a dominating set, all nodes in the compatibility gadget $I^{i,j}$ are also dominated. Since $w_i$ dominates all nodes in $I^{i,j}$ except one, we must have that $w_j$ in $K^j$ dominates the remaining node in $I^{i,j}$. But this implies by construction that $w_i$ and $w_j$ correspond to different nodes in $G$, and they are also not neighbours in $G$. Thus, $D$ corresponds to an independent set of size $k$ in $G$.
\end{itemize}
 
\paragraph{\boldmath Simulation in the congested clique.} It remains to argue that given an input graph $G$ and a dominating set algorithm $A$ with running time $O(n^\delta)$, we can simulate in the congested clique the execution of $A$ on $G'$ in $O(n^\delta)$ rounds. We have each node $v \in V$ simulate the nodes $v^i$ and $v^{i,j}$ for the possible choices of parameters $i$ and $j$; by construction, $v$ can determine all edges incident to those nodes in $G'$ from its local view of $G$. Moreover, we have nodes $1$ and $2$ simulate the special nodes in $x^i$ and $y^i$, respectively. Thus, each node is simulating at most $O(k^2)$ nodes in $G'$. The running time of $A$ in $G'$ is $O((k^2n)^\delta)$, and the overhead from simulating $O(k^2)$ nodes per each node in $G$ is $O(k^4)$ rounds for each round in $G'$, for total running time $O(k^{2\delta+4}n^\delta)$.

\subsection{Vertex cover and fixed-parameter tractability}\label{sec:vc-ub}

\paragraph{\boldmath Vertex cover.} Minimum vertex cover and maximum independent set are known to be essentially the same problem. However, we show that the \emph{parameterised} version of vertex cover is significantly easier than independent set in the congested clique model:

\begin{theorem}
Vertex cover of size $k$ can be found in $O(k)$ rounds in the congested clique model.
\end{theorem}

We use a very simple idea that also appears in the centralised Buss kernelisation algorithm~\cite{buss-kernel,pa-book}. 

\begin{lemma}\label{lemma:vc}
If $G = (V,E)$ has a vertex cover $C$ of size $k$, and $v \in V$ is a vertex of degree at least $k+1$, then $v \in C$.
\end{lemma}

\begin{proof}
If $v \notin C$, then all neighbours of $v$ are in $C$, which implies $|C| \ge k+1$.
\end{proof}

In the congested clique, we exploit this idea as follows. As a preprocessing phase, all nodes of degree at least $k+1$ join the vertex cover---denote this set by $C$---and broadcast this information to all nodes; if more than $k$ nodes attempt to join the vertex cover, we know that there is no vertex cover of size $k$. As the main phase, all nodes $v \notin C$ broadcast full information about their incident edges not covered by $C$, and all nodes locally compute a minimum vertex cover for $G[V \setminus C]$. By \lemmaref{lemma:vc}, we know that there is a vertex cover of size $k$ for $G$ if and only if there is a vertex cover of size $k - |C|$ for $G[V \setminus C]$. 

The preprocessing phase takes a single round. Since all nodes of degree at least $k + 1$ join $C$ in the preprocessing phase, all nodes active in the main phase have to broadcast information about at most $k$ edges, so the main phase takes at most $k$ rounds.

\paragraph{\boldmath Fixed-parameter tractability.} The above result, taken together with prior upper bounds results for parameterised problems such as $k$-independent set and $k$-cycle detection, illustrate that ideas from \emph{fixed-parameter algorithms}~\cite{pa-book, downey2012parameterized} can also be applied in the congested clique. Specifically, the following is known in the congested clique model:
\begin{itemize}
    \item A vertex cover of size $k$ can be found in $O(k)$ rounds -- the complexity is dependent polynomially on $k$, and not at all on $n$.
    \item A $k$-path can be found in $\exp(k)$ rounds~\cite{congest-subgraphs,disc2017property} and a $k$-cycle in $\exp(k)n^{0.157}$ rounds~\cite{censor2015algebraic} -- the complexity is exponential in $k$, but the complexity in terms of $n$ is independent of $k$.
    \item An independent set of size $k$ can be found in $O(n^{1-2/k})$ rounds and a dominating set of size $k$ in $O(n^{1-1/k})$ rounds -- the complexity in terms of $n$ is dependent on $k$.
\end{itemize}
Compare this to what we know from the centralised setting:
\begin{itemize}
    \item Vertex cover, $k$-path and $k$-cycle are all \emph{fixed-parameter tractable} problems: they can be solved in time $\exp(k)\operatorname{poly}(n)$. However, vertex cover admits a \emph{polynomial kernel}~\cite{buss-kernel,pa-book} -- intuitively, this means that vertex cover instances can be compressed to size $\operatorname{poly}(k)$ in polynomial time\footnote{We refer to the recent book by Cygan et al.~\cite{pa-book} for the formal definition of a kernel, as well as detailed discussion of the topic.} -- while the other two do not~\cite{bodlaender2009problems}.
    \item By contrast, parameterised versions of independent set and dominating set are W[1]-hard and W[2]-hard, respectively, which strongly suggests that they are not fixed-parameter tractable, but rather require $n^{\Omega(k)}$ time to be solved~\cite{downey2012parameterized}.
\end{itemize}

\section{Conclusions}\label{sec:conclusions}

\paragraph{\boldmath Nondeterminism.} As a major open question, we highlight the lack of separation between constant-round deterministic and nondeterministic congested clique. It seems reasonable to conjecture that
\[ \clique(1) \ne \nclique(1)\,,\]
but it is not clear how we should approach proving such separation. Indeed, it would be interesting even if we could prove this conditional on a centralised complexity assumption, such as $\text{P}\ne\text{NP}$.

\paragraph{\boldmath $\nclique(1)$ as an \lcl{} analogue.} As we noted before, the class $\nclique(1)$ plays a similar role in the congested clique to the class \lcl{} of locally checkable labellings that been in the focus of recent complexity-theoretic work in the \local{} model. Note that here we refer to \lcl{} problems in the original sense of Naor and Stockmeyer~\cite{naor95what}, that is, class \lcl{} consists of search problems such as 2-colouring, sinkless orientation and maximal independent set, where a valid output can be verified in constant rounds. By contrast, work on local decision often uses \lcl{} to refer to the corresponding labelling verification problems, such as verifying a valid colouring of a graph~\cite{fraigniaud11ld,goos16lcp}.

Similarly to the \lcl{} problems, the class of $\nclique(1)$-labelling problems on labelled graphs is a natural class of search problems on the congested clique. We define an $\nclique(1)$-labelling problem $L$ as a set of pairs $(G,z)$, where $G$ is an input graph with edge and node labels of $O(\log n)$ bits, $z$ is an output labelling, and the membership $(G,z) \in L$ is decidable in constant rounds. Given an input graph, the computational task is to output a label $z_v$ for each node $v$ such that $(G,z) \in L$, or reject if such labelling does not exist. This class captures many natural graph problems of interest, but we do not have lower bounds for any problem in this class.

\paragraph{\boldmath Randomness.} In this work, we have focused on deterministic and nondeterministic computation; however, there are problems in the congested clique model where the best \emph{randomised} upper bounds are significantly better than the best deterministic upper bounds, e.g.\ minimum spanning tree~\cite{lotker05,ghaffari2016complexity}. Thus, a possible extension of the present work is to study the randomised complexity landscape of the congested clique. Indeed, the counting arguments of Applebaum et al.~\cite{doi:10.1142/S0129626416500043} extend to randomised protocols. Likewise, \theoremref{thm:nclique-hierarchy} implies that there are problems that cannot be solved in $O(S(n))$ rounds with one-sided Monte Carlo algorithms, but can be solved in $O(T(n))$ rounds deterministically for $S(n) = o(T(n))$, as the Monte Carlo algorithm can be converted to a nondeterministic algorithm.

\medskip

\paragraph{\boldmath Acknowledgements.} This work was supported in part by the Acad\-emy of Finland, Grant 285721. We thank Alkida Balliu, Parinya Chalermsook, Magn\'{u}s M.\ Halld\'{o}rsson, Juho Hirvonen, Petteri Kaski, Dennis Olivetti and Christopher Purcell for comments and discussions.

\DeclareUrlCommand{\Doi}{\urlstyle{same}}
\renewcommand{\doi}[1]{\href{http://dx.doi.org/#1}{\footnotesize\sf doi:\Doi{#1}}}
\bibliographystyle{plainnat}
\bibliography{clique-complexity}

\end{document}